\documentclass[%
 reprint,
 amsmath,amssymb,
 aps,
]{revtex4-1}

\usepackage{natbib}
\usepackage{graphicx}% Include figure files
\usepackage{dcolumn}% Align table columns on decimal point
\usepackage{bm}% bold math
\usepackage{mathtools}
\usepackage{subfigure}
\usepackage{amsthm}
\usepackage{amsmath}
\usepackage{amssymb}

\newtheorem*{thm}{Theorem}

\begin{document}

\preprint{APS/123-QED}

\title{Self-Avoiding Pruning Random Walk on Signed Network}% Force line breaks with \\

\author{Huijuan Wang}

\author{Cunquan Qu}%
 \email{c.qu-1@tudelft.nl; cqq890708@gmail.com}
\altaffiliation[Also at ]{%
 School of Mathematics, Shandong University, Jinan, 250100, China.
}%

\author{Chongze Jiao}
\author{Wioletta Rusze}
\affiliation{Faculty of Electrical Engineering, Mathematics, and Computer Science, Delft University of Technology, Mekelweg 4, Delft, The Netherlands, 2628 CD.}%Lines break automatically or can be forced with \\%

\date{\today}% It is always \today, today,

\begin{abstract}
A signed network represents how a set of nodes are connected by two logically contradictory types of links: positive and negative links. In a signed products network, two products can be complementary (purchased together) or substitutable (purchased instead of each other). Such contradictory types of links may play dramatically different roles in the spreading process of information, opinion, behavior etc. In this work, we propose a Self-Avoiding Pruning (SAP) random walk on a signed network to model e.g. a user's purchase activity on a signed products network. A SAP walk starts at a random node. At each step, the walker moves to a positive neighbour that is randomly selected and its previously visited node together with its negative neighbours are removed. We explored both analytically and numerically how signed network topological features influence the key performance of a SAP walk: the evolution of the pruned network resulted from the node removals, the length of a SAP walk and the visiting probability of each node. These findings in signed network models are further verified in two real-world signed networks. Our findings may inspire the design of recommender systems regarding how recommendations and competitions may influence consumers' purchases and products' popularity. 
% \begin{description}
% Secondary publications and information retrieval purposes.
% \item[PACS numbers]
% May be entered using the \verb+\pacs{#1}+ command.
% You may use the \texttt{description} environment to structure your abstract;
% use the optional argument of the \verb+\item+ command to give the category of each item. 
% \end{description}
\end{abstract}

\pacs{Valid PACS appear here}% PACS, the Physics and Astronomy
\maketitle

%\tableofcontents

\section{\label{sec:intro}Introduction}
The concept of multi-layer networks has been proposed in 2010 \cite{Interdependent_Nature2010,kivela2014multilayer,
PhysRevX_mathMultilayer,sahneh2013generalized,Cardillo2013} to capture different types of relationships/links among the same set of nodes. For example, the rapid development of the Internet, smart phones and information technology has facilitated the boost of online platforms, such as Facebook and YouTube, for communications, creating and sharing information and knowledge. Users may participate in one or several online networks besides their physical contacts forming a multi-layer network where the nodes represent the users and the links in each layer represent a specific type of connections such as physical contacts and online follower-followee relationships. Such multi-layer networks support the spreading of e.g. information, behavioural patterns, opinions, fashion within each layer respectively and allow as well these spreading processes on different layers to interact, introducing new phenomena that dramatically differ from a single spreading process on a single network \cite{de2018author,de2016physics,sole2016random,PhysRevE.86.026106,6580178,wang2013effect,
cozzo2013contact,DaqingLi,li2014epidemics,competingprocess,dynamicinterplay,Li2013,
de2017disease,1367-2630-19-7-073039}.

Signed networks is a special type of two-layer networks where the same set of nodes are connected by two logically contradictory types of links, so called positive and negative links. The positive and negative links may represent friendly and antagonistic interactions respectively in a signed social network \cite{Leskovec:2010:SNS:1753326.1753532} and represent the complementary (i.e. when a product e.g. a phone is purchased, the other product e.g. a phone charger is likely to be bought in addition) and substitutable (two products can be purchased instead of each other such as the phones from two competing brands) relationships respectively in a signed network of products \cite{McAuley:2015:INS:2783258.2783381,Chiang:2014:PCS:2627435.2638573,DBLP:journals/corr/TangCAL15}.

Whereas all types of links in most multi-layer networks such as physical contact and online friendships are mostly positive thus facilitate the spread of information, opinion and etc., the positive and negative links in a signed network usually play dramatically different roles in a general spreading process. Hence, we propose, in this work a Self-Avoiding Pruning (SAP) Walk on a signed network to model, e.g.  a user's purchase behaviour on a signed network of products.
\begin{figure}
\centering
    \includegraphics[width=3.5in]{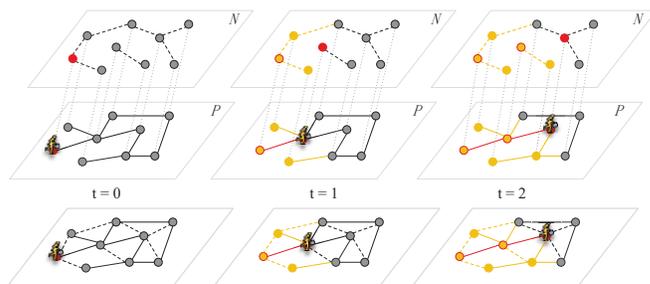}
\caption{Schematic plot of SAP Walk on a signed network. The signed network is represented on the top as a two-layer network£¬ a Negative layer and a Positive layer where dashed lines between the two layers emphasise that the nodes are the same individual across layers and represented at the bottom as a single network with two types of links: positive (solid lines) and negative (dotted lines). At $t=0$, a walker visits a random node in the network, which is in red. At step 1 or $t=1$, the walker moves to a random positive neighbour (in red) and the previously visited node and its negative neighbours (in yellow) are removed. Such steps repeat until the walker has no node to visit any more. The pruned network in grey is shrinking over time.}
\label{fig:sapwalk}
\end{figure}
As shown in Figure~\ref{fig:sapwalk}, a SAP walk starts at a random node in a signed network at $t=0$. At each step, the walker moves from its current location node $i$ to a positive neighbour\footnote {A positive (negative) neighbour of a node is a neighbour that is connected with this node via a positive (negative) link.} $j$ that is randomly selected and its previous location, i.e. node $i$ as well as node $i$'s negative neighbours are removed from the signed network\footnote{When a node is removed from the network, all links connected to the node, including both positive and negative links, are also removed.}. The walker repeats such steps until there is no new location to move to.  Since each node pair can be connected by either a positive or negative link, but not both, the walker could equivalently, at each step, removes firstly the negative neighbours of its current visiting node $i$, then moves to a random positive neighbour $j$ and afterwards removes the previously visited node $i$.

In the context of a signed product network, a SAP walk may model the purchase trajectory of a user on the network of products: initially, the user purchases a random product and afterwards buys a random complementary product of his/her previous purchase; however, the user will not buy the same product repetitively nor the substitutable products of what he/she has bought. When a user buys a product, the complementary products are recommended to a user in online shopping platforms like Amazon. The random purchase of a random complementary product has been assumed or modelled in the past (see Section \ref{relatedwork}) \cite{selfavoidwalk,rwproduct,masuda2017random}.

In the scenario of a signed social network, a SAP walk may model how an opinion/policy adopted (rejected) by supportive (antagonistic) partners. Consider the scenario of a signed social network, where the nodes could be either individuals or companies. A SAP walk could model the trajectory that a walker spreads/lobbies for his opinion, technology or business solutions: an individual/company tends to be easier to be convinced by the walker if one of his friends or collaborating companies has been reached by the walker, which motivates a walker to preach hop by hop through the positive network layer. However, an individual tends to be difficult to convince if any of his antagonistic partners has been reached by the walker, which motivates the walker to avoid of the negative neighbours of those that have been visited.

In this paper, we aim to understand how features of a signed network influence:  (1) the evolution of the pruned network topology resulted from the node removals in a SAP walk (2) the length/hopcount of a SAP walk, i.e. the number of positive links that a SAP walker traverses in total and (3) the visiting probability of each node by a SAP walk. Taking the product network as an example, we are actually going to explore how the initial signed product network features influence (1) the subnetwork a user may further explore after several purchase actions and (2) the probability that a product is purchased and (3) how many purchases a user may perform in total in a SAP walk. The signed product network thus may affect both the user purchase behaviour and the popularity or market share of the products. These questions will be firstly explored on simple signed network models and afterwards on two real-world signed networks.

Our observations and analysis point out the significant influence of the negative/substitutable links on the number of purchases of a user, the distribution of the popularity of a product etc. These findings may inspire the design of future recommender systems: which complementary product(s) should be recommended in order to maximize the total number of purchases of a walker? With which products a product should not compete in order to maximise the total purchases of the product? How competition between products may affect market share of the products and users' purchases?

The SAP model can be improved or extended, especially as more rich data becomes available. The complementary products of a product are not necessarily recommended with the same strength or priority. Section \ref{generalisation} illustrates one possible generalisation of the SAP model where a substitutable product of a purchased product is removed with a given probability.

The paper is organised as follows: we introduce the basic definitions related to signed network models and random walks in Sec. \ref{definition}. The influence of the signed network features on the aforementioned properties are studied in signed network models in Sec. \ref{prune}, Sec. \ref{length} and \ref{visitingprobability}, respectively. The influence of the correlation between the positive and negative degrees of a node on these SAP walk properties is discussion in Sec. \ref{degreecorrelation}. Our observations and understanding obtained in signed network models are further verified in two real-world signed networks in Sec \ref{realworld}. We summarise our findings and discuss promising future work is in Sec. \ref{conclusion}.

\section{Definitions}
\label{definition}
In this section, we introduce basic definitions regarding to signed network representation and models, different types of random walks\cite{masuda2017random} and their relation to the SAP walk.

\subsection{Signed Network Representation}
In a signed network with $N$ nodes, two $N \times N$ adjacency matrices $A^+$ and $A^-$
can be used to represent the positive and negative connections respectively. Element $A^+_{i,j}=1$ ($A^-_{i,j}=1$) if node $i$ and $j$ are connected via a positive (negative) link. Or else $A^+_{i,j}=0$ ($A^-_{i,j}=0$). A signed network can be as well represented by a single adjacency matrix $A=A^+-A^-$. Hence, each element
\begin{displaymath}
A_{i,j} =
\begin{cases}
1, \text {if $i$ and $j$ are connected via a positive link}  \\
-1, \text{if $i$ and $j$ are connected via a negative link} \\
0,  \text {if $i$ and $j$ are not connected}
\end{cases}
\end{displaymath}

The positive degree $d^+_{i}= \sum_{j=1}^{N}A^+_{i,j}$ of a node $i$ counts the number of positive neighbours of node $i$, whereas the negative degree $d^-_{i}= \sum_{j=1}^{N}A^-_{i,j}$ indicates the number of negative links incident to a node $i$. As shown in Section \ref{realworld}, both the positive and negative degree of a node in real-world signed networks tend to follow a power law distribution. Since each node pair can be connected only by one type of links, positive of negative, but not both, $A^-_{i,j}A^+_{i,j}=0$. An example of signed network is shown in Figure~\ref{fig:sapwalk}, which is plotted both as a two-layer network (above) and a single network with two types of links (bottom).

\subsection{Signed Network Models}

The simplest signed networks can be constructed by generating the positive layer and negative layer independently from the same network model or two different network models respectively, such as the Erd\H{o}s-R\'{e}nyi ER and scale-free SF random network model.

Erd\H{o}s-R\'{e}nyi ER random network is one of the most studied random network
models that allow many problems to be treated analytically
\cite{ERgraph1,ERrandom}. To generate an Erd\H{o}s-R\'{e}nyi random network
with $N$ nodes and average degree $E[D]$, we start with $N$ nodes and place
each link between two nodes that are chosen at random among the $N$
nodes until a total number $L=\frac{N * E[D]}{2}$ of links have been placed.
All the links are bidirectional. In this paper, we choose $N=1000$, $E[D]=4$ for
Erd\H{o}s-R\'{e}nyi random networks. Erd\H{o}s-R\'{e}nyi random networks are
characterised by a Poisson degree distribution, $Pr[D=k] = \frac{(Np)^{k}%
e^{-Np}}{k!}$, where $D$ is the degree of a random node in the network, and the link density $p=\frac{E[D]}{N-1}$.

\noindent We use the hidden parameter model \cite{PhysRevLett.89.258702,PhysRevE.66.066121,PhysRevE.68.036112,BAnetsci2016} to generate scale-free
networks which have a power-law degree distribution $Pr[D=k]=ck^{-\lambda}$
as observed in many real-world networks \cite{configmodel3,config1,BAnetwork}. The hidden parameter model is considered because the degree distribution and the average degree of the generated scale-free networks are both controllable. We start with $N$ isolated nodes and assign each node $i$ a hidden parameter $\eta_{i}=\frac{1}{i^{\alpha}}$, $i=1, 2, ..., N$. At each step, two nodes $i$ and $j$ are chosen randomly with a probability proportional to $\eta_{i}$ and $\eta_{j}$ and they are connected as a link if they were not connected previously. Such steps are repeated until $L=\frac{E[D]N}{2}$ links have been added. In this case, the generated random network has a power-law degree distribution $Pr[D=k]=ck^{-(1+\frac{1}{\alpha})}$. In this paper, we consider $N=1000$, average degree $E[D]=4$ and $\lambda=3$, such that the ER and SF networks have the same average degree and a size of the largest connect component close to $N$.

The positive and negative degree of a node are possibly correlated, actually positively correlated as shown in the real-world signed networks in Section \ref{realworld}. Moreover, triangles with an odd number of positive links, so called balanced triangles, have been shown to appear more frequently than the other types of signed triangles \cite{Szell03082010}.

We focus on the simplest signed networks where the positive and negative connections are generated independently from either the same or different network models, i.e. ER or SF model. In this case, the positive and negative degree of a node are uncorrelated. We construct 4 types of signed networks: ER-ER, ER-SF, SF-SF and SF-ER, where $N=1000$ and  the average degree of both layers are $4$. Moreover, we consider as well ER-ER networks where  $E[D^+]=4$ and $E[D^-]=0, 4$ and $8$ to explore the influence of the density of the negative layer on SAP walks.

We explored as well the signed networks ER-ER and SF-SF where the positive and negative degree of a node are positively correlated with linear correlation coefficient $0 \leq \rho \leq 1 $, $E[D^+] = E[D^-] = 4$ and $N=1000$. Such networks are generated as follows. First, an ER (or SF) network is generated as the positive network layer. Second, set the negative degree of each node the same as its positive degree. Third,  select randomly a fraction $1-\rho$ of the nodes and shuffle randomly their negative degrees. After the shuffling, the generated degree sequences for the two layers are correlated with linear correlation coefficient $\rho$ \cite{Shlomo_correlation, PhysRevE.90.052811}. Given the negative degree of each node, construct the negative network layer according to the configuration model \cite{Newmanconfiguration}.

\section{Related Work}
\label{relatedwork}
Classic random walk (RW) starts at a random node in an unsigned network. At each step, the walker moves from its current location node $i$ to a neighbour that is selected uniformly at random. In this process, the walker can visit any node repeatedly if the network is connected. Random walk has been widely applied e.g. to model network routing protocol, users' visit at websites via hyper links and to detect network topology \cite{7539632,5466694,6688590,6406279,7737565}.
The Self-Avoiding random walk (SAW) is the same as the random walk except that at each step the walker moves to a random neighbour that has not yet been visited. Hence, each node can be visited maximally once. A SAW stops when the walker has no further node to visit any more. The SAW was first introduced by chemist P. J. Froly to study the behaviour of polymers on lattice graph \cite{flory1953principles}. SAW has also been applied to detect protein-protein interaction\cite{7035059}, to detect network structure which is more efficient than classic random walk by avoiding previously visited nodes in each step \cite{2014-Camilleri-p4791-4791}, and to detect unidentified network traffic\cite{nia2016detecting}.

Performance of these two types of random walks have been analytically studied \cite{lawler1980}. The probability that a node is visited by a classic random walker has been shown to be proportional to the degree of that node. The path length of a SAW is the number of links that has been traversed in total in a SAW. The path length of SAW has been studied, especially regarding to the average and the probability distribution \cite{2012arXiv1212.3448G}. I. Tishby, O. Biham and E. Katzav have found that the path length of SAW on an Erd\H{o}s-R\'{e}nyi random network follows the Gompertz distribution in the tail \cite{1751-8121-49-28-285002}.

RW and SAW have been used to model users' purchase activity on a recommendation network of products where two products are connected if two products are frequently co-prochased by the same users \cite{selfavoidwalk,rwproduct}.
In contrast to RW and SAW, SAP walk addresses further that products can be substitutable to each other and are seldom or not purchased by the same user.

Jung et al. considered signed random walk, where the sign of the walker changes depending on the signs of the links that walker has traversed \cite{walksigned}. This work addresses, for the first time, that the signed links could influence the dynamics of the walk thus the walkers' trajectories. The SAP walk is more challenging to trace analytically compared to other random walk models that have been studied so far. It is equivalent to a self-avoiding walk on the positive network layer if the negative network layer is empty, i.e. no negative links exist.

Opinion diffusion (voter model) on a signed network has been proposed in \cite{diffusionsigned}. Dynamics of influence diffusion and influence maximisation problem on signed networks have been explored \cite{diffusionsigned} beyond the influence maximisation problem on single unsigned networks \cite{influencemaxi}. Viral spreading processes on signed networks have been studied in \cite{epi_signed}. In this paper, we address the self-avoiding random walks on Signed networks.  This is more challenging to address analytically because the earlier trajectory of a walker influences its future moves, in contrast to spreading models where state transitions of each node depends only on the current states of neighbours and the local dynamic rule.

\section{Evolution of the pruned network structure}

\label{prune}
The Self-Avoiding Pruning (SAP) Walk on a signed network is more complex than previous random walk models. At each step, the walker moves to a random positive neighbour, and afterwards, not only the previously visited node but also its negative neighbours are removed/pruned from the signed network. As shown in Figure ~\ref{fig:sapwalk}, the pruned signed network (in grey) $G(t)$ is shrinking over time. The pruned signed network at a step $t$ refers to the remaining signed network after the removal of nodes at step $t$. The initial signed network corresponds to $G(0)$.

The pruned positive network layer $G^{+}(t)$ suggests the potential subgraph of the original signed network that the walker could further explore via a SAP walk. In this section,  we will explore how the topology, especially the average degree, of the pruned positive network layer is changing over time. We start with the simpler case when both the initial positive and negative layers are ER networks, with possibly different average degree.

Firstly, we examine the case when the initial negative layer is an empty graph, i.e. the average degree $E[D^{-}(t=0) = 0]$ initially is zero. The initial positive network layer possesses the binomial degree distribution, which approaches the Poisson distribution $Pr[D^{+}(0)=k]=\frac{c(0)^k}{k!}e^{-c(0)}$, under the condition that the network is sparse, i.e. the average degree $c(0)=E[D^{+}(0)]$ of the initial positive layer at $t=0$ is a constant. The SAP walk on such a signed network is equivalent to a SAW on the positive network layer $G^{+}(0)$. Initially, the network has $N(t=0)=N$ nodes. At any step $t$, the pruned positive network layer has $N-t$ nodes.

An insightful observation of a SAW walk on an ER random graph in \cite{1751-8121-49-28-285002} is as follows. A SAW walker has a higher probability to visit a neighbour with a higher degree. Take the step $t=1$ as an example. Starting from the random node that is visited at step $t=0$, the walker walks to a node with degree $k$ in $G^{+}(0)$ at step $t=1$ with probability $kPr[D^{+}(0) = k]/c(0)$. A special property of the Poisson distribution is that $kPr[D^{+}(0) = k]/c(0)=Pr[D^{+}(0) = k-1]$. The probability to walk to a node with degree $k$ in $G^{+}(0)$, thus with degree $k-1$ in $G^{+}(1)$ after the removal of the previously visited node and its links is $Pr[D^{+}(0) = k-1]$. The node to visit at $t=1$ is as if chosen randomly from $G^{+}(1)$.

Note that when a randomly selected node together with its links are removed from an ER network, the remaining network is again an ER network with the same link density, i.e. the probability that two nodes are connected. Hence, the network pruning resulting from a SAP walk on an ER positive network with an empty negative layer is statistically equivalent to the node removal process upon the initial ER network where at each step, a node is randomly selected and this node together with its links are removed from the network. The pruned positive network  $G^{+}(t)$ at any step $t$ is thus an ER network with $N-t$ nodes and link density $E[D^{+}(0)]/(N-1)$. The average degree \footnote{The average degree of the pruned positive or negative network layer at a step $t$ refers to  the average degree at step $t$ over all the nodes and over a large number of SAP walk realisations.} of the pruned positive network at step $t$ is
\begin{equation}E[D^{+}(t)]=\frac{E[D^{+}(0)]}{N-1}(N-t-1)
\end{equation}
when the original signed network is sparse and the size $N$ is large.

Furthermore, we consider the case where the negative network layer $G^{-}(0)$ is not empty but an ER network with its average degree $E[D^{-}(0)]>0$.

\begin{thm} When the positive and negative layers of a signed network are large, sparse ER networks and independently generated with link density $p^{+}=E[D^{+}(0)]/(N-1)$ and $p^{-}=E[D^{-}(0)]/(N-1)$ respectively, the average degree $E[D^{+}(t)]$ of the pruned positive network layer at time $t$ of a SAP walk follows
\begin{multline}
\label{equ_ER_avgD}
E[D^{+}(t)] \approx (N-1+\frac{1}{p^{-}})p^{+}(1-p^{-}) ^{t}-\frac{p^{+}}{p^{-}}
\end{multline}
\end{thm}

\begin{proof}At any step $t$ of a SAP walk, network is pruned by the removal of the previously visited node and its negative neighbours. Since the negative and positive ER networks are generated independently, the node to visit at each step is as if chosen randomly from the negative layer $G^{-}(t)$. From the view of the negative network layer, a random node and its negative neighbours together with all their negative links are removed from the negative layer at each step. The negative network layer $G^{-}(t)$ remains approximately an ER network with the same link density over time. This is an approximation because the neighbour of a random node tends to have a higher degree. The link density remains approximately the same  $p^{-}=E[D^{-}(0)]/(N-1)$ over time, whereas at each step $t$ on average, $N(t-1)-N(t)=1+(N(t-1)-1)p^{-}$ nodes are removed. The size $N(t)$ refers to the average size of the pruned network at step $t$ over a large number realisations of the stochastic SAP walks. Hence, the size of the pruned negative layer, which is as well the size of the pruned positive layer follows
\begin{equation}
N(t)\approx (N-1+\frac{1}{p^{-}})(1-p^{-}) ^{t}-\frac{1}{p^{-}}+1
\end{equation}
From the prospective of the positive layer, at step $t$, the negative neighbours of the previously visited node are as if chosen randomly from positive layer $G^{+}(t-1)$. The positive layer remains as an ER network with the same link density $p^{+}=E[D^{+}(0)]/(N-1)$. The average degree of the pruned positive network layer at time $t$ is $E[D^{+}(t)]=(N(t)-1)p^{+}$, which leads to \ref{equ_ER_avgD}.
\end{proof}

When the original signed network is sparse and the size $N$ is large,
\begin{equation}
E[D^{+}(t)]\simeq(N+\frac{1}{p^{-}})p^{+}(1-p^{-}) ^{t}-\frac{p^{+}}{p^{-}}
\end{equation}
 \begin{figure}
  \centering
    \includegraphics[width=3.5 in]{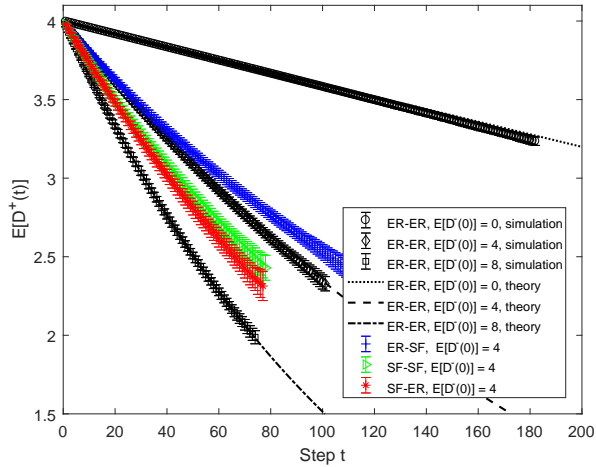}
  \caption{The average degree $E[D^{+}(t)]$ of the positive pruned network as a function of the step $t$ of a SAP walk on a signed network with $N=1000$ nodes. We simulate 100 independent realisations of a SAP walk on each of the independently generated 100 signed networks and plot the average and standard deviation (error bar) of $E[D^{+}(t)]$ . The average degree of the initial positive layer is always $E[D^{+}(0)]=4$. The SF network has an power exponent 3. }
  \label{fig:avgD} %% label for entire figure
\end{figure}

As shown in Figure~\ref{fig:avgD}, the average degree $E[D^{+}(t)]$ of the positive pruned network as a function of the SAP walk step can be well approximated by our theoretical result Equation (\ref{equ_ER_avgD}), when the initial network is an ER-ER signed network. The average positive degree $E[D^{+}(t)]$ of the pruned network at each step $t$ is proportional to the size of the pruned network as shown in Equation (\ref{equ_ER_avgD}). A dense negative layer with a large $E[D^{-}(0)]$, could more effectively prune the network, leading to a fast decrease of the average degree of the pruned positive layer.

ER(positive)-SF(negative) signed networks are pruned less than ER-ER networks when both layers have the same average degree $4$ (see Figure~\ref{fig:avgD}). This can be explained as follows. If a visited node has a large negative degree, its removal will lead to the removal of many nodes, its negative neighbours. If a negative neighbour of a visited node has a large negative degree, however, the removal of such a negative neighbour together with its negative links will not remove extra nodes but makes the negative layer sparser, protecting the network from the pruning. In ER-SF networks, nodes with a high negative degree in the SF negative layer are likely to be removed as the negative neighbour of a visited node, which reduces the pruning.
Figure~\ref{fig:avgD} shows that the pruned positive network, e.g. $E[D^{+}(t)]$, shrinks faster if the initial network is a SF-SF signed network than ER-SF signed network. This is mainly due to the fact that, a node with a large positive degree is likely to be visited in early steps and removed, significantly reducing the average degree of the positive pruned layer. However, the negative neighbours of a visited node are as if chosen randomly in the positive layer and tend to have a lower positive degree in a SF-SF network than in ER-SF network, slightly reducing the pruning effect.

Hence, a SF or in general a heterogeneous positive layer and a dense negative layer tend to facilitate the pruning of the network whereas a SF (heterogeneous) negative layer reduces the pruning effect.

\section{Length of a SAP walk}
\label{length}
The length or hopcount $H$ of a SAP walk counts the total number of positive links, or the total number of move steps, a SAP walker traverses until it has no other node to move to. In the context of a signed produce network, $H+1$ suggests the total number of purchases of a consumer. Signed networks leading to a large $H+1$ promotes the purchases of more products. We would like to understand how the original signed network topology influences the length of a SAP walk.
\begin{figure}
  \centering
    \includegraphics[width=3.5in]{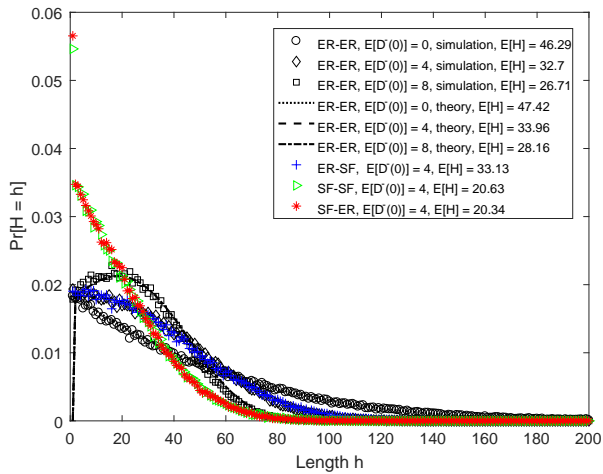}
  \caption{Probability distribution of the length of a SAP walk on a signed network.}
  \label{fig:pathlength} %% label for entire figure
\end{figure}

The probability distribution of the length $H$ of a SAP walk is shown in Figure~\ref{fig:pathlength}, for various types of signed networks. Intuitively, a SAP walk stops when the walker has no other node to move to, which is likely to happen if the current pruned positive network layer is less connected in the sense that no giant connected component exists but only small connected clusters exist. Hence, a dense initial negative network layer $G^{-}(0)$ leads to the removal of many nodes in each step and effectively reduces the connectivity of the positive layer, resulting in a small length $H$. This explains our observation in ER-ER signed network (Figure ~\ref{fig:pathlength}) that a dense $G^{-}(0)$ leads to a small length $H$ on average. %In Figure~\ref{fig:pathlength}, we also find that the length $H$ is smaller on average in a SF-SF network than that in an ER-ER network, when both networks have the same average degree in both layers. Indeed, a heterogeneous initial negative layer (e.g. a SF network) has a high chance that a node with a high positive degree is visited by a SAP in early steps and removed, effectively reducing the connectivity of the pruned positive network as well as the length $H$ of a SAP walk. %

The distribution of the length $H$ of a SAP walk on an ER-ER network can be  analytically derived. A SAP walk has a length $H=h$ requires that the node that the walker visits at step $h$ has degree $0$ in the pruned positive layer $G^{+}(h)$ and each node visited in a previous step $t$ where $0\leq t<h$, has a positive degree in the corresponding pruned positive network layer larger than $0$. As discussed in Section \ref{length}, the pruned positive layer remains an ER network with the same link density $p^+$ but with a shrinking size $N(t)$. Hence,
\begin{equation}\label{eq_length}
Pr[H=h]= (1-p^+)^{N(h)-1}\prod\limits_{t=0}^{h-1}[1-(1-p^+)^{N(t)-1}].
\end{equation}

Figure~\ref{fig:pathlength} shows that our theory of the distribution of the length Equation~\ref{eq_length} well approximates the simulation results.

Which type of signed networks tend to lead to a long length of a SAP walk? If we look at the average path length, the ordering of signed networks from the highest to the lowest follows: $ER-ER(E[D^{-}(0)=0])$, $ER-SF$, $ER-ER(E[D^{-}(0)=4])$,$ER-ER(E[D^{-}(0)=8])$,$SF-SF$,$SF-ER(E[D^{-}(0)=4])$. This ordering is consistent with our previous explanation: a heterogeneous positive layer such as SF network and a dense negative layer facilitate the pruning of the network leading to a short length of a SAP walk whereas a heterogeneous e.g. SF negative layer reduces the pruning effect attributing to a long length of a SAP walk.

The length of a SAP walk actually depends on not only the link density but as well the connectivity of the pruned positive layer. The negative neighbours of a visited node are as if chosen randomly in the positive layer since the two layers are independent in connections. Removal of such random nodes in the positive layer reduces less the connectivity of the positive layer if the original positive layer is a SF network since SF networks are robust against random node removals compared to ER networks. However, hubs in the positive layer are more likely to be visited and removed reducing more significantly the density of the SF positive layer.

\section{Nodal Visiting Probability}

\label{visitingprobability}
The probability that a node is visited by a SAP walk implies a certain kind of importance of the node, e.g. the probability that a product is purchased when the signed network represents the network of products. Intuitively, a node with a higher positive degree in the initial signed network has a higher chance to be visited by a SAP walk. Hence, we examine the visiting probability of a node given its initial positive degree, which is shown in Figure ~\ref{fig:visitprob}.

\begin{figure}
  \centering
\subfigure[]{
    \includegraphics[width=3.5in]{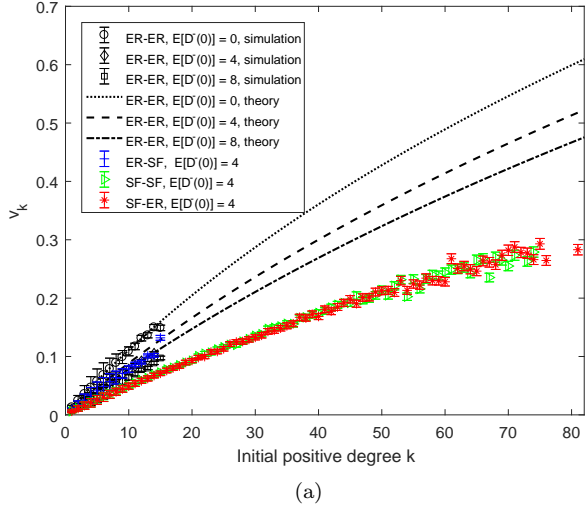}}
\hspace{1in}
\subfigure[]{
    \includegraphics[width=3.5in]{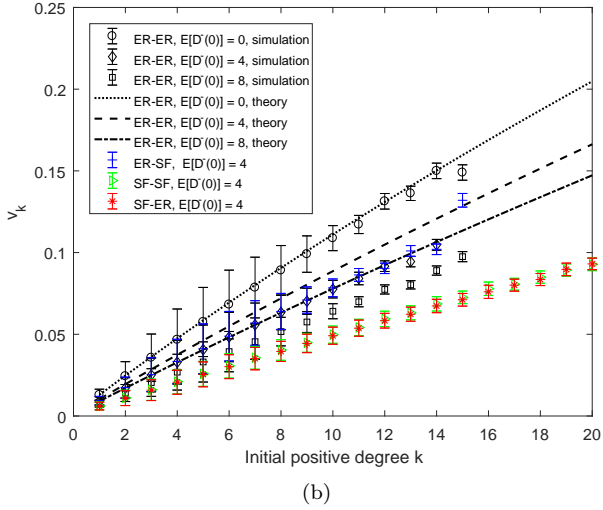}}
  \caption{The probability $v_k$ that a node is visited by a SAP walk v.s. its degree $D^+(0)=k$ in the initial positive network is plotted for all possible degree values $k$ in (a) and for small $k$ in (b). Since many nodes may have the same initial positive degree $D^+(0)=k$, we plot their average visiting probability and the standard deviation.}
  \label{fig:visitprob} %% label for entire figure
\end{figure}

Firstly, we analytically derive the nodal visiting probability in ER-ER networks. Specifically, we compute the probability $v_k$ that a random node $j$ with degree $d^+_j(0)=k$ in the initial positive layer is visited by a SAP  walk starting at a random node. We denote $X_t$ as the node that is visited by a SAP walk at step $t$. Since the node $j$ can be visited at any step $0\leq t\leq h$ by a SAP walk of length $h$, we have
\begin{equation}
\begin{aligned}
v_k=\sum\limits_{h=0}^{N(0)-1} \sum\limits_{t=0}^{h}Pr[X_t=j|H=h]Pr[H=h] \\
=\sum\limits_{t=0}^{N(0)-1}Pr[X_t=j|H\geq t]Pr[H\geq t],
\end{aligned}
\end{equation}
assuming that the probability node $j$ is visited at any step $t$ is independent of the length $H$ of the walk as long as $H\geq t$. The node $j$ is visited at step $t$ by a SAP walk that has a length $H\geq t$ requires that node $j$ is not visited nor removed in the previous steps and $j$ is connected with the node $X_{t-1}$ visited at step $t-1$. Hence,
\begin{equation}
\begin{aligned}
Pr[X_t=j|H\geq t]=\prod\limits_{t'=0}^{t-1}(1-Pr[X_{t'}=j]) \cdot  \\
\frac{\Delta_j^+(t-1)}{N(t-1)-1} \cdot \frac{1}{p^+(N(t-1)-1)},
\end{aligned}
\end{equation}
where $Pr[X_0=j] = 1/N(0)$, $\Delta_j^+(t-1)$ is the degree of node $j$ at step $t-1$ in the pruned positive network layer given that $j$ is not visited in the first $t-1$ steps. The node $X_t$ to be visited at step $t$ as well as its negative neighbours are as if randomly chosen from the pruned positive layer $G_{t-1}^+$. The pruned network remains proximately (precisely if the negative layer is empty) an ER network with the same link density $p^+$ when the node visited and its negative neighbours are removed at each step. The ratio $\frac{\Delta_j^+(t-1)}{N(t-1)-1}$ is the probability that node $j$ is connected with the node $X_{t-1}$ visited in the previous step and $\frac{1}{p^+(N(t-1)-1)}$ is the probability that the walker choose node $j$ out of the $p^+(N(t-1)-1)$ positive neighors of $X_{t-1}$ to move to. We approximate the degree $\Delta_j^+(t')$ by its average using the same symbol, which follows the following recursion for $t'<t$ thus before the node is visited
\begin{equation}
\begin{aligned}
\Delta_j^+(t') =& \frac{\Delta_j^+(t'-1)}{N(t'-1)-1} (\Delta_j^+(t'-1)-p^-(\Delta_j^+(t'-1)-1)-1 )       \\
&+ (1-\frac{\Delta_j^+(t'-1)}{N(t'-1)-1}) (\Delta_j^+(t'-1)-p^-\Delta_j^+(t'-1))\\
=& \Delta_j^+(t'-1)(1-p^-)\frac{N(t'-1)-2}{N(t'-1)-1}
\label{equ:degreekrecur}
\end{aligned}
\end{equation}
where $\Delta_j^+(0)=D_j^+(0) = k$. The first (second) term corresponds to the case that node $j$ is (not) connected with the node visited at step $t'-1$. In the first case where $j$ is connected with $X_{t'-1}$, the degree $\Delta_j^+(t')$ at step $t'$ could be reduced from $\Delta_j^+(t'-1)$ due to the removal of $X_{t'-1}$ and its negative neighbours which happen to be a positive neighbour of node $j$.  In the second case, the degree $\Delta_j^+(t')$ decreases from $\Delta_j^+(t'-1)$ due to the removal of $X_{t'-1}$'s negative neighbours which happen to be a positive neighbour of node $j$. Combining Equation (2) and (5)-(8), we could derive the probability $v_k$ that a random node $j$ with degree $d^+_j(0)=k$ in the initial positive layer is visited by a SAP walk on an ER-ER signed network.

As shown in Figure ~\ref{fig:visitprob}, our numerical solution of nodal visiting probability well approximates the simulation results especially when the initial negative ER network is sparse e.g. $E[D^{-}(0)]=0$. When the initial negative ER network is denser, the actual visiting probability is lower than predication of the numerical solution. This is because our theoretical analysis assumes that the negative layer remains an ER network with the same link density after the removal of each visited node and its negative neighbours, as if all these nodes removed are chosen randomly. In fact, high negative degree nodes are more likely to be removed as a negative neighbour of a visited node. The actual $\Delta_j^+(t')$, thus also the visiting probability, is smaller than their corresponding analytical estimations.

Interestingly, the order of the various types of signed networks in the heterogeneity of nodal visiting probabilities, i.e. the slope, from the highest to the lowest, the order of these networks in the average SAP walk length and the order of the signed networks in the average degree of the pruned positive layer at a given step are the same. A SAP walk that prunes the network slowly tends to have a long length and lead to a high heterogeneity in nodal visiting probabilities. A SAP walker tends to visit high degree nodes in the pruned positive layer at each step. A longer length of a SAP walk, thus, attributes to a higher visit probability of a node with a large initial positive degree, leading to more heterogeneity of nodal visiting probabilities.
\section{Influence of degree-degree correlation}
\label{degreecorrelation}
The positive and negative degree of a node can be correlated. The real-world networks considered in Section \ref{realworld} have a positive correlation between the degrees of a node in the two layers. In this section, we explore how the degree-degree correlation between the positive and negative layers may influence the aforementioned performance of a SAP walk.

We consider ER-ER and SF-SF signed networks with $E[D^+] = E[D^-] = 4$ and $N=1000$, where the degree-degree correlation $\rho$ varies within $[0,1]$. We illustrate the three properties of a SAP walk using the two extreme case $\rho=0$ and $\rho=1$, whereas results from other $\rho$ values within $[0,1]$ lead to the same observations. We simulate 100 independent realisations of a SAP walk on each of the independently generated 100 signed networks to derive the three properties of SAP walks.

In both ER-ER and SF-SF signed networks, we find that a positive degree-degree correlation evidently facilitates the pruning of the network, reduces the average path length and leads to a more homogeneous visiting probabilities among the nodes (see Figure ~\ref{fig:avgD_corr}, ~\ref{fig:pathlength_corr}, and ~\ref{fig:visitingprob_corr} ). Such effects are more evident in SF-SF networks than in ER-ER networks and can be explained as follows. When the degree-degree correlation is positive, a high degree node in the positive layer tends to have a high degree in the negative layer. The high positive degree of such a node tends to let the walker visit the node in earlier steps. After being visited, the node together with its many negative neighbours, are removed, pruning the network significantly. The high negative degree of such a node tends to let the node be removed as a negative neighbor of a node that has been visited. The removal of such a node together with its many positive links significantly prunes the positive layer and reduces the connectivity.
\begin{figure}
  \centering
    \includegraphics[width=3.5 in]{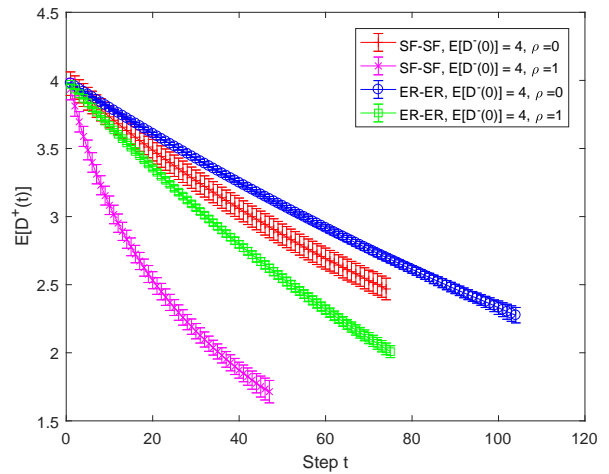}
  \caption{The average degree $E[D^{+}(t)]$ of the positive pruned network as a function of the step $t$ of a SAP walk on a signed network when the positive and negative degrees of a node are uncorrelated $\rho=0$ and correlated $\rho=1$ respectively.}
  \label{fig:avgD_corr} %% label for entire figure
\end{figure}
\begin{figure}
  \centering
    \includegraphics[width=3.5in]{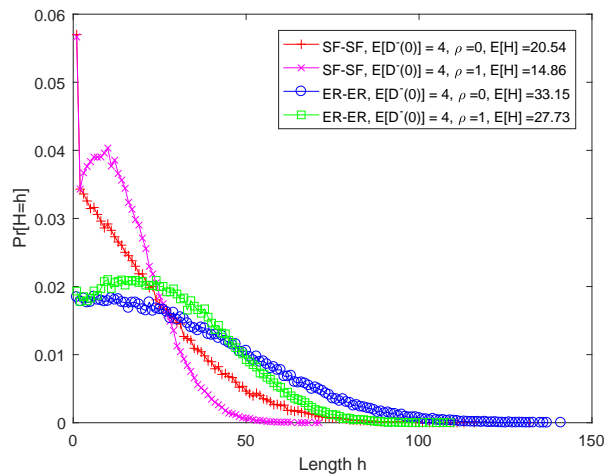}
  \caption{Probability distribution of the length of a SAP walk on a signed network when the positive and negative degrees of a node are uncorrelated $\rho=0$ and correlated $\rho=1$ respectively.}
  \label{fig:pathlength_corr} %% label for entire figure
\end{figure}
\begin{figure}
  \centering
    \includegraphics[width=3.5in]{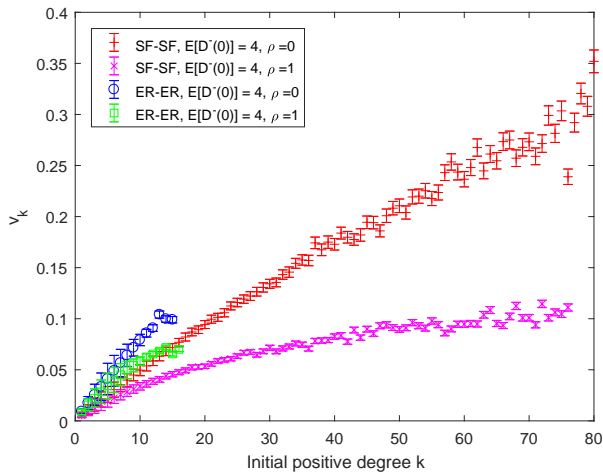}
  \caption{The probability $v_k$ that a node is visited by a SAP walk v.s. its degree $D^+(0)=k$ in the initial positive network, when the positive and negative degrees of a node are uncorrelated $\rho=0$ and correlated $\rho=1$ respectively.}
  \label{fig:visitingprob_corr} %% label for entire figure
\end{figure}

\section{Generalisation of the SAP walk model}
\label{generalisation}

The SAP walk model can be generalised from multiple perspectives to better approximate real-world purchase behavior of users. As more rich data recording user online activities becomes available, we may discover the possibly heterogenous preference of a user over the recommended products. The positive neighbours of a node that has just been visited may have different probabilities to be selected by the walker to move to.

It is possible, though small in chance, that a user buys two substitutable products. The negative neighbours of a node that have just been visited are not necessarily removed in reality, but could be removed with a given probability or with different probabilities. We illustrate one extension of the SAP walk: once a node is visited, each of its negative neighbours is removed independently with a pruning probability $r$. We first consider the ER-ER and SF-SF networks where the positive and negative layers are generated independently.
\begin{figure}
  \centering
    \includegraphics[width=3.5in]{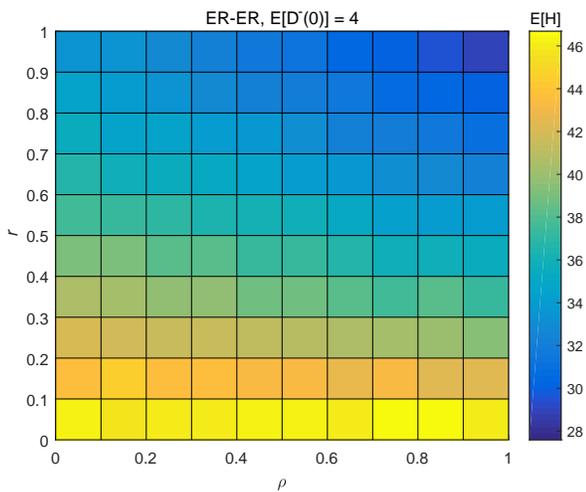}
  \caption{Phase diagram of the average path length $E[H]$ of generalised SAP walks. Dependence of the average path length on the pruning probability $r$ and the degree-degree correlation $\rho$. The color bar in the right represents the average path length. The underlying signed network is an ER-ER network with $E[D^{+}(0)]=E[D^{-}(0)]=4$. For each set of parameters, the average path length is obtained as the average of 1000 SAP walks on each of the 100 signed networks.}
  \label{fig:hop_r_corr_ER} %% label for entire figure
\end{figure}
\begin{figure}
  \centering
    \includegraphics[width=3.5in]{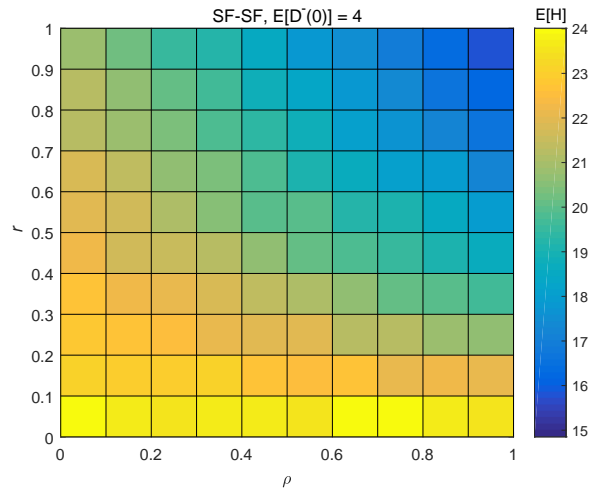}
  \caption{Phase diagram of the average path length $E[H]$ of generalised SAP walks. Dependence of the average path length on the pruning probability $r$ and the degree-degree correlation $\rho$. The color bar in the right represents the average path length. The underlying signed network is an SF-SF network with $E[D^{+}(0)]=E[D^{-}(0)]=4$. For each set of parameters, the average path length is obtained as the average of 1000 SAP walks on each of the 100 signed networks.}
  \label{fig:hop_r_corr_SF} %% label for entire figure
\end{figure}
The extended SAP walk with a pruning probability $r$ on an ER-ER signed network with an average degree $E[D^+]$ and $E[D^-]$ for the two layers respectively is equivalent to the SAP walk with pruning probability $1$ on an ER-ER signed network with an average degree $E[D^+]$ and $E[D^-]*r$ respectively. Adding a pruning probability $r$ to a SAP walk model is equivalent to scaling the link density $p^{-}=\frac{[D^-]}{N-1}$ of the negative layer to $p^{-}*r$.

However, such equivalence does not hold when the positive and negative degrees of a node are correlated. We take the average path length of a generalised SAP walk as an example and explore the effect of the pruning probability $p$ and the degree-degree correlation $\rho$ on the average path length of a generalised SAP walk. As shown in Figure~\ref{fig:hop_r_corr_ER} and ~\ref{fig:hop_r_corr_SF}, the effect of the pruning probability on the average hopcount is more evident as the degree-degree correlation increases. When the degree-degree correlation is high, nodes with a high degree in both layers tend to be removed in early steps of a walk. In this case, a smaller pruning probability could effectively reduce the pruning.

\section{SAP walks on real-world signed networks}
\label{realworld}
Finally, we choose two real-world signed networks and explore their network features and how these features may influence the SAP walks on these networks.
We consider the Wikipedia adminship election network and an Extracted Epinions social network \cite{snapnets}. In Wiki network, two nodes connected by a positive (negative) link suggest that the two users support (reject) each other to be an administrator. A positive (negative) link in Epinions network means that the corresponding two users trust (distrust) each other¡¯s reviews.

The Epinions network is far larger than Wiki. We have sampled the Epinions network by firstly removing all nodes with zero positive degree or zero negative degree and then randomly selecting the same number of nodes as in Wiki from the largest connected positive layer of Epinions together with the positive and negative links among these nodes. Basic topological features of these two networks of the same size are showed in Table. 1.  The degree correlation $\rho _D$ measures the linear correlation coefficient between the positive degree and negative degree of a node. The positive and negative layers tend to be positively correlated in their degrees, i.e. $\rho _D>0$, instead of independent as assumed in our signed network models.

\begin{table*}
\begin{center}
\begin{tabular}{ |c|c|c|c|c|c|c|c| }
\hline
Network& Nodes & Links & fraction of "+"Links & fraction of "-"Links& E[$D_+$] & E[$D_-$] & $\rho _D$\\
\hline
Wiki & 6186& 97874 & 78402 (80.11\%) & 19472 (19.89\%) & 25.35 & 6.30 & 0.62 \\
\hline
Epinions & 6186 & 92091 & 73062 (79.33\%) & 19029 (20.67\%) & 23.62 & 6.15 & 0.35 \\
\hline
\end{tabular}
\end{center}
\caption{Topological Parameters of Real Networks}\label{table:1}
\end{table*}
The degree distributions of the positive and negative layer in both Wiki and Epinions are shown in Figure ~\ref{fig:degree_real} are highly heterogeneous, closer to a scale-free distribution than a Poisson distribution.
\begin{figure}
  \centering
  \subfigure[Wiki network]{
    \includegraphics[width=3.0in]{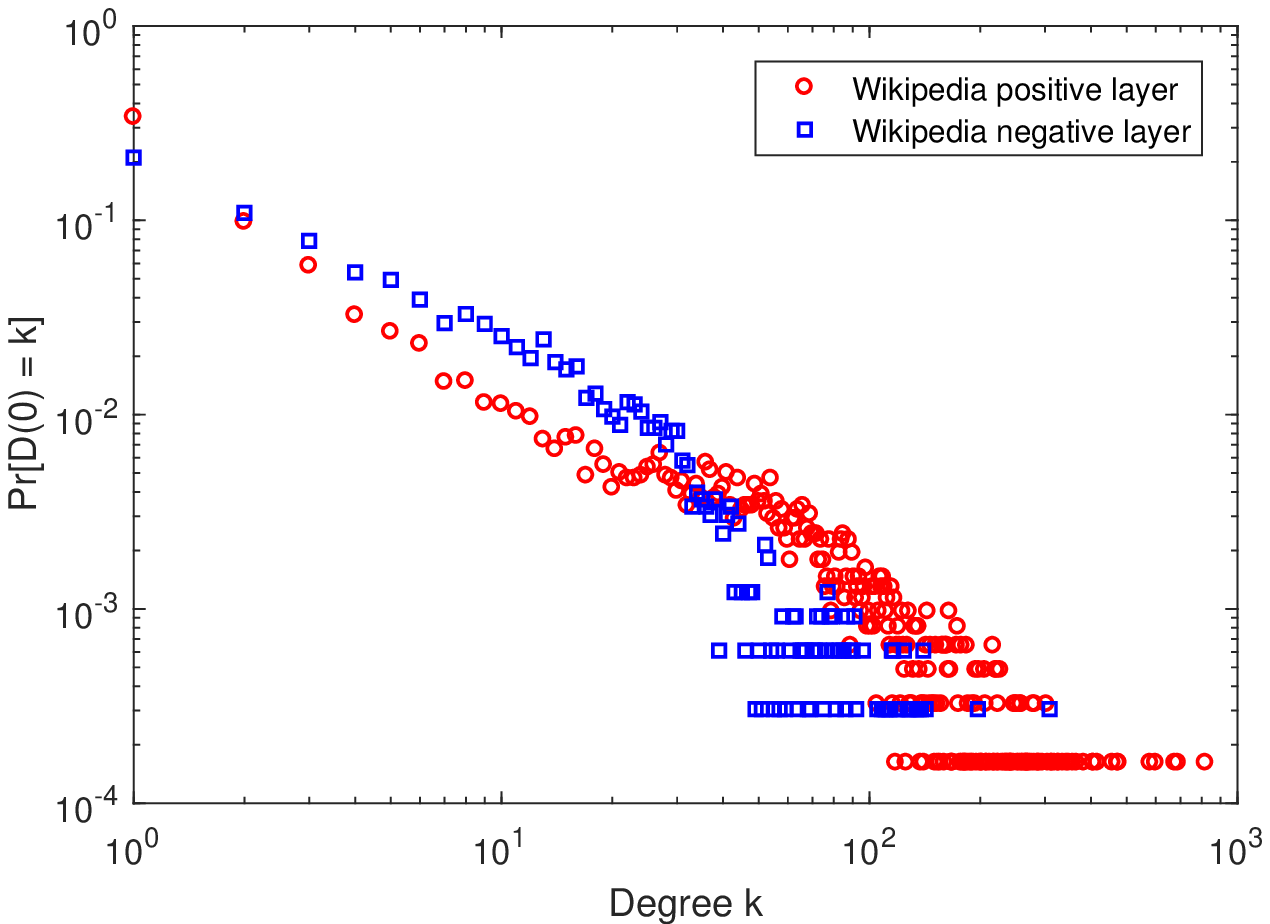}}
  \hspace{0.2in}
    \subfigure[Epinions network]{
    \label{fig:subfig:c} %% label for second subfigure
    \includegraphics[width=3.0in]{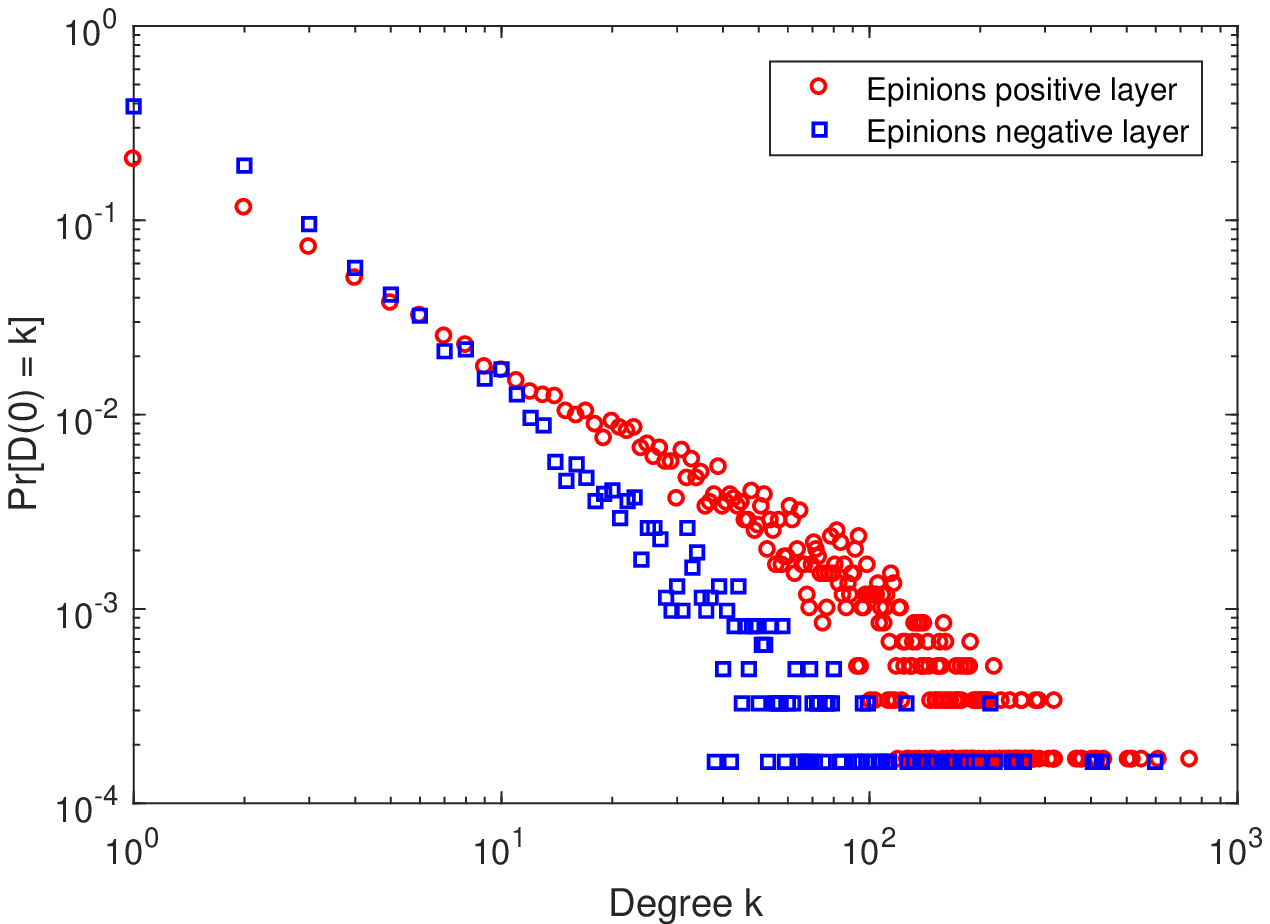}}
  \caption{Degree distribution of real-world signed networks.}
  \label{fig:degree_real} %% label for entire figure
\end{figure}

Upon each real-world signed network, we simulate independently $10^{5}$ SAP walks and investigate their key properties discussed earlier. The negative network layer has been missing in modelling the purchase behavior of a user. Hence, we consider as well the SAP walk on these two real-world networks, where, however, the negative network layer is replaced as an empty network without any link.

\begin{figure}
  \centering

    \includegraphics[width=3.5in]{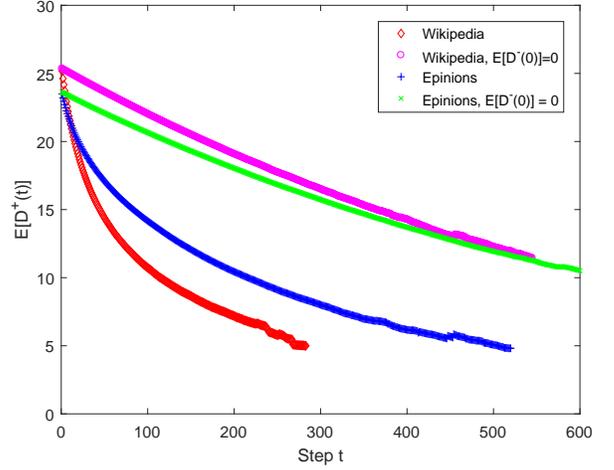}

  \caption{The average degree $E[D^{+}(t)]$ of the positive pruned network as a function of the step $t$ of a SAP walk on a real-world signed network.}
  \label{fig:avgD_real} %% label for entire figure
\end{figure}
Figure ~\ref{fig:avgD_real} shows that the positive layer is pruned or shrinks faster in Wiki than in Epinions network. Wiki has a shorter length of SAP walk on average than Epinions, as shown in Figure ~\ref{fig:length_real}. One explanation for both observations is that Wiki has a slightly denser initial negative layer (larger $E[D^{-}(0)]$) than Epinions as shown in Table 1, which removes on average more nodes per step. Moreover, the high degree correlation $\rho_D$ in Wiki contributes as well to a fast pruning in the positive layer and a short length of SAP walks.

When the negative layer is empty, i.e. $E[D^{-}(0)=0]$, the positive layer is pruned far slower, the average path length $E[H]$ is far larger. A SAP walk on a signed network with an empty negative layer is equivalent to a self-avoidance walk on the positive network layer. In this case, the Epinion positive layer leads to a longer average path length than the Wiki positive layer. This is likely due to the higher standard deviation of the degree in the Wiki positive layer $49.55$ than that in the Epinion $45.58$ positive layer. A higher degree standard deviation implies an earlier visit of the hubs, whose removal may significantly prune the network and reduce the connectivity.

\begin{figure}
  \centering

    \includegraphics[width=3.5in]{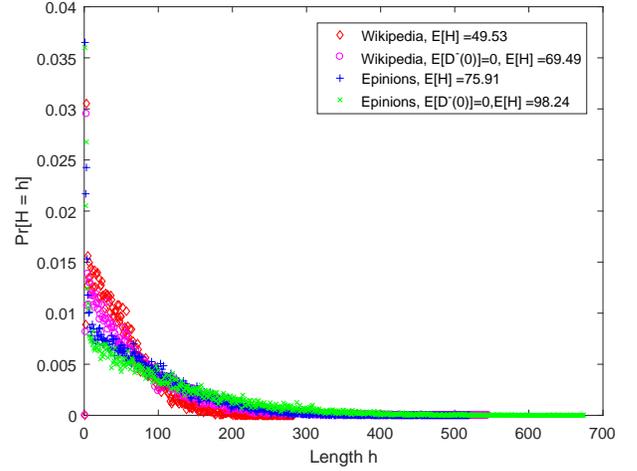}

  \caption{Probability distribution of the length of a SAP walk on a signed real-world network.}
  \label{fig:length_real} %% label for entire figure
\end{figure}

\begin{figure}
  \centering

    \includegraphics[width=3.5in]{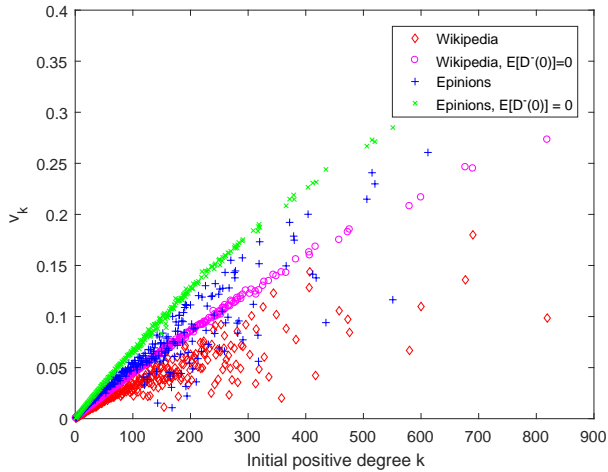}

  \caption{The probability $v_k$ that a node is visited by a SAP walk v.s. its degree $D^+(0)=k$ in the initial positive network.}
  \label{fig:visitprb_real} %% label for entire figure
\end{figure}

The visiting probability of a node versus its initial positive degree tends to have a larger slope in Epinions than that in Wiki. This is consistent with our observations in signed network models that the visiting probability $v_k$ increases faster with $k$ in a signed network that leads to a higher average degree $E[D^{+}(t)]$ of the pruned positive layer and a longer length $E[H]$ of SAP walks.

The negative network layer dramatically prunes a signed network and reduces length of a SAP walk. Moreover, a heterogeneous degree distribution in the positive layer and a positive degree-degree correlation between positive and negative layers may further enhance the pruning effect, shorten the SAP walk length and facilitate homogeneous visiting probabilities of nodes. These effects have been observed consistently in both network models and real-world networks.

\section{CONCLUSION}
\label{conclusion}
Classic spreading models assume that all network links are beneficial for information diffusion. However, the positive and negative links in a signed network may facilitate and prevent the contagion of information, opinion and behavior etc. respectively. As a start, we propose a Self-Avoiding Pruning (SAP) Random Walk on a signed network to model, for example, a user's purchase activity on a signed network of products and information/opinion diffusion on a signed social network. We unravel the significant effect of the negative links and the signed network structure in general on SAP walks. We found that a more heterogeneous degree distribution of the positive network layer such as the power-law distribution, a denser negative layer and a high degree-degree correlation between the two layers tend to prune the network faster, suppress the length of SAP walks and reduce the heterogeneity in nodal visiting probabilities. When the two layers are independent, however, a more heterogeneous degree distribution of the negative network layer tends to slow down the pruning and contribute to a longer length of SAP walks and more heterogeneity in nodal visiting probabilities. These observations has been obtained from both signed network models and real-world signed network and analytically proved in signed ER-ER networks. Real-world networks tend to have a heterogeneous degree distribution in the positive layer and a positive degree-degree correlation, which reduce total purchases of users but increase the homogeneity of the popularity of products. Our findings point out the possibility to influence users' purchases and product popularity via recommendations and competitions.

Beyond the basic features of signed networks considered in this work, it is interesting to explore further, how other key features such as the fraction of balanced triangles, which appear more frequently than unbalanced ones in real-world networks, affect SAP walks and other dynamic processes in general. We could as well to improve the SAP walk towards a more realistic model of e.g. user's purchase activity, by taking into account, for example, the choice of the initial node to visit, the possibility that a walker/user may stop the walk earlier and the heterogeneity of the links preference over recommendations. Optimisation problems that are interesting to be further explored include how to add nodes to an existing signed network and how to add positive links via e.g. recommendations to maximise visiting probabilities of a group of nodes.

\section*{Acknowledgment}
The authors would like to thank National Nature Science Foundation of China (Nos. 11601430, 11631014, 11471193), the Foundation for Distinguished Young
Scholars of Shandong Province (JQ201501) for support.

\bibliography{reference}
\end{document}